\documentclass[letterpaper]{IEEEtran}
\ifCLASSINFOpdf
\else
\fi
\usepackage{flushend}
\usepackage{stfloats}
\usepackage{subfigure}
\usepackage{fancyhdr}
\usepackage{mathrsfs}
\usepackage{graphicx}
\usepackage{geometry}
\usepackage{color}
\usepackage{placeins}
\usepackage{float}
\usepackage{tabularx,colortbl}
\usepackage{amsmath}
\usepackage{amsfonts}
\usepackage{amssymb}
\usepackage{amsthm}
\usepackage{bm}
\usepackage[numbers,sort&compress]{natbib}
\begin{document}
%
% paper title
% can use linebreaks \\ within to get better formatting as desired
\title{Wideband mmWave Channel Estimation for Hybrid Massive MIMO with Low-Precision ADCs}
\author{\IEEEauthorblockN{Yucheng~Wang\IEEEauthorrefmark{0}, Wei~Xu\IEEEauthorrefmark{0}, Hua~Zhang\IEEEauthorrefmark{0}, and Xiaohu~You\IEEEauthorrefmark{0}}
\vspace{-1cm}
\thanks{Manuscript received July 20, 2018; accepted September 10, 2018. This work was supported by the NSFC under grants 61871109, 61601115, and 61571118, the Six Talent Peaks Project in Jiangsu Province under GDZB-005, and the Open Research Fund of the State Key Laboratory of ISN under Grant ISN18-03. The editor coordinating the review of this paper and approving it for publication was S. K. Mohammed. \emph{(Corresponding author: Wei Xu.)}

Y. Wang and W. Xu are with the National Mobile Communications Research Laboratory (NCRL), Southeast University, Nanjing, China, and also with the State Key Laboratory of Integrated Services Networks, Xidian University, Xi'an, China (\{yc.wang, wxu\}@seu.edu.cn).

H. Zhang and X. You are with the NCRL, Southeast University, Nanjing, China (\{huazhang, xhyu\}@seu.edu.cn).}
}
% author names and affiliations
% use a multiple column layout for up to three different
% affiliations

% conference papers do not typically use \thanks and this command
% is locked out in conference mode. If really needed, such as for
% the acknowledgment of grants, issue a \IEEEoverridecommandlockouts
% after \documentclass

% for over three affiliations, or if they all won't fit within the width
% of the page, use this alternative format:
%
%\author{\IEEEauthorblockN{Michael Shell\IEEEauthorrefmark{1},
%Homer Simpson\IEEEauthorrefmark{2},
%James Kirk\IEEEauthorrefmark{3},
%Montgomery Scott\IEEEauthorrefmark{3} and
%Eldon Tyrell\IEEEauthorrefmark{4}}
%\IEEEauthorblockA{\IEEEauthorrefmark{1}School of Electrical and Computer Engineering\\
%Georgia Institute of Technology,
%Atlanta, Georgia 30332--0250\\ Email: see http://www.michaelshell.org/contact.html}
%\IEEEauthorblockA{\IEEEauthorrefmark{2}Twentieth Century Fox, Springfield, USA\\
%Email: homer@thesimpsons.com}
%\IEEEauthorblockA{\IEEEauthorrefmark{3}Starfleet Academy, San Francisco, California 96678-2391\\
%Telephone: (800) 555--1212, Fax: (888) 555--1212}
%\IEEEauthorblockA{\IEEEauthorrefmark{4}Tyrell Inc., 123 Replicant Street, Los Angeles, California 90210--4321}}

% use for special paper notices
%\IEEEspecialpapernotice{(Invited Paper)}

% make the title area
\maketitle
\thispagestyle{fancy}
\renewcommand{\headrulewidth}{0pt}
\pagestyle{fancy}
\cfoot{}
\rhead{\thepage}

\newtheorem{mylemma}{Lemma}
\newtheorem{mytheorem}{Theorem}
\newtheorem{mypro}{Proposition}
\begin{abstract}
%\boldmath
In this article, we investigate channel estimation for wideband millimeter-wave (mmWave) massive multiple-input multiple-output (MIMO) under hybrid architecture with low-precision analog-to-digital converters (ADCs). To design channel estimation for the hybrid structure, both analog processing components and frequency-selective digital combiners need to be optimized. The proposed channel estimator follows the typical linear-minimum-mean-square-error (LMMSE) structure and applies for an arbitrary channel model. Moreover, for sparsity channels as in mmWave, the proposed estimator performs more efficiently by incorporating orthogonal matching pursuit (OMP) to mitigate quantization noise caused by low-precision ADCs. Consequently, the proposed estimator outperforms conventional ones as demonstrated by computer simulation results.
\end{abstract}
% IEEEtran.cls defaults to using nonbold math in the Abstract.
% This preserves the distinction between vectors and scalars. However,
% if the conference you are submitting to favors bold math in the abstract,
% then you can use LaTeX's standard command \boldmath at the very start
% of the abstract to achieve this. Many IEEE journals/conferences frown on
% math in the abstract anyway.

% keywords

% For peer review papers, you can put extra information on the cover
% page as needed:
% \ifCLASSOPTIONpeerreview
% \begin{center} \bfseries EDICS Category: 3-BBND \end{center}
% \fi
%
% For peerreview papers, this IEEEtran command inserts a page break and
% creates the second title. It will be ignored for other modes.
\IEEEpeerreviewmaketitle

\vspace{-0.5cm}
\section{Introduction}
\vspace{-0.2cm}
Millimeter wave (mmWave) massive multiple-input multiple-output (MIMO) is an emerging technology for future wireless networks. Typical massive MIMO is equipped with a large number of radio frequency (RF) chains, which are cost- and power-hungry, especially for wideband mmWave systems. Hybrid architecture with limited RF chains recently attracts much attention for mmWave massive MIMO to reduce cost and complexity \cite{LeLiangWCL}. However, it imposes additional challenges for channel estimation because fully digital processing is no longer accessible. In \cite{ghauch2015subspace}, a subspace-based channel estimator has been presented for narrowband massive MIMO with the hybrid structure. In \cite{venugopal2017channel}, the channel sparsity has been further utilized and an orthogonal matching pursuit (OMP)-based least-square (LS) estimator has been proposed. A distributed grid matching pursuit (DGMP) algorithm \cite{dailinlong2016CL} has been proposed to solve the power leakage in uplink channel estimation for mmWave MIMO. Further in \cite{heath2018JSTSP}, a simultaneous weighted-OMP estimator has been developed.

For wideband mmWave MIMO, high-precision analog-to-digital converters (ADCs) are expensive and power-hungry \cite{walden1999analog}. In order to alleviate the burden, low-precision~ADCs have been introduced. However, channel estimate is deteriorated due to the nonlinear quantization of low-precision ADCs. A linear-minimum-mean-square-error (LMMSE) estimator has been developed in \cite{li2017channel} for massive MIMO with 1-bit ADCs.

In this article, we investigate channel estimation for mmWave massive MIMO with the hybrid architecture using low-precision ADCs. Main contributions are summarized as
\begin{itemize}
\item The problem is reformulated into a statistically equivalent linear estimation problem which optimizes both frequency-flat analog processing weights and frequency-selective digital ones in sequence.
\item Statistics of the equivalent noise are derived. Consequently the optimal digital estimator is obtained in closed form.
\item The proposed hybrid estimator applies for arbitrary channel models. If \emph{a priori} sparsity presents, the proposed estimator incorporating OMP can potentially further mitigate quantization noise caused by low-precision ADCs.
\end{itemize}

The rest of this paper is organized as follows. In Section II, we introduce the system model. The channel estimation algorithm is presented in Section III. Simulation results and conclusions are presented in Sections IV and V, respectively.

\emph{Notations}: $\mathbf{A}^{\rm H}$, $\mathbf{A}^{\rm T}$, and $\mathbf{A}^*$ are the conjugate transpose, transpose, and conjugate of $\mathbf{A}$, respectively. ${\rm vec}(\mathbf{A})$ and ${\rm diag}(\mathbf{A})$, respectively, return vectorization and the diagonal matrix containing diagonal elements of $\mathbf{A}$. Operator $\otimes$ represents the Kronecker product. $\mathbb{E}\{\mathbf{A}\}$, ${\rm Tr}(\mathbf{A})$, and $[\mathbf{A}]_{ij}$ are the expectation, trace, and $(i,j)$th element of $\mathbf{A}$, respectively. $\|\mathbf{a}\|_p$ is the $l_p$-norm of vector $\mathbf{a}$, and $\lceil a\rceil$ is the ceiling function of scalar $a$. $\mathcal{CN}(0,1)$ indicates circularly symmetric complex Gaussian distribution with mean 0 and variance 1.

\begin{figure}
\centering
\includegraphics[width=3.2in]{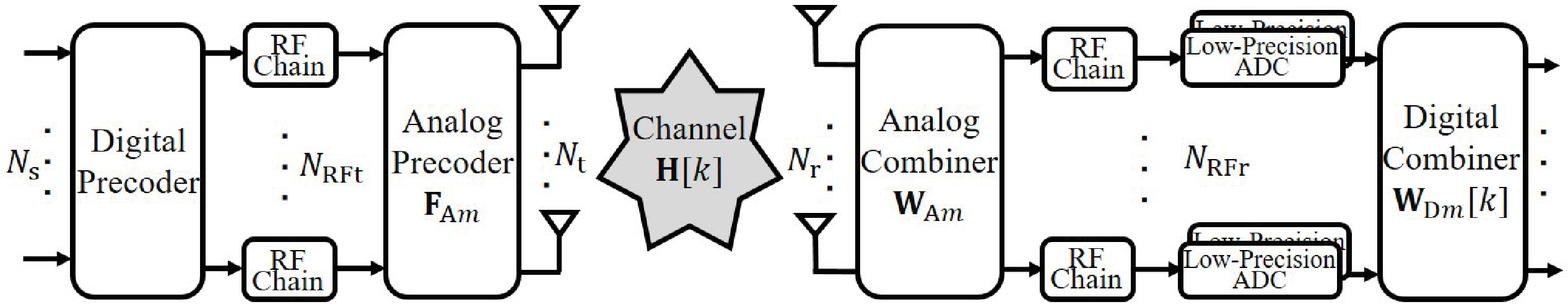}
\vspace{-0.3cm}
\caption{Block diagram.}
\label{Fig:system_model}
\vspace{-0.8cm}
\end{figure}

\vspace{-0.3cm}
\section{System Model}
\vspace{-0.2cm}

We consider a wideband mmWave massive MIMO-OFDM system with low-precision ADCs where hybrid precoder and combiner are used. The transmitter is equipped with $N_{\rm t}$ antennas driven by a smaller number, $N_{\rm RFt}$, of RF chains. The receiver is equipped with $N_{\rm r}$ antennas and $N_{\rm RFr}$ RF chains.

Channel estimation in OFDM is usually performed in the frequency domain. Let $N_{\rm c}$ and $K$ be the maximal delay tap and the number of subcarriers, respectively. The channel response matrix of the $d$th path can be expressed as \cite{alkhateeb2014channel}
\begin{equation}\label{eq:Hd}
\vspace{-0.2cm}
\mathbf{H}_d=\sum_{l=1}^{N_{\rm p}}\alpha_l\delta(d - \tau_l)\mathbf{a}_{\rm R}(\theta_{{\rm R}l})\mathbf{a}_{\rm T}^{\rm H}(\theta_{{\rm T}l}),
\vspace{-0.1cm}
\end{equation}
for $d=0, 1, \cdots, N_{\rm c}-1$, where $\alpha_l$ is the channel gain of the $l$-th path, $\tau_l$ is the normalized path delay, and $\mathbf{a}_{\rm R}(\theta_{{\rm R}l})\in\mathbb{C}^{N_{\rm r}\times1}$ and $\mathbf{a}_{\rm T}(\theta_{{\rm T}l})\in\mathbb{C}^{N_{\rm t}\times1}$ are antenna array response vectors under uniform linear array (ULA) setup at the receiver and the transmitter, respectively.
\vspace{-0.2cm}
\begin{equation}
\mathbf{a}_{\rm R}(\theta_{{\rm R}l})=\frac{1}{\sqrt{N_{\rm r}}}\big[1, e^{-j2\pi\theta_{{\rm R}l}}, \cdots, e^{-j2\pi(N_{\rm r}-1)\theta_{{\rm R}l}}\big]^{\rm T},
\vspace{-0.2cm}
\end{equation}
%\begin{equation}
%\mathbf{a}_{\rm T}(\theta_{{\rm T}l})=\frac{1}{\sqrt{N_{\rm t}}}\big[1, e^{-j2\pi\theta_{{\rm T}l}}, \cdots, e^{-j2\pi(N_{\rm t}-1)\theta_{{\rm T}l}}\big]^{\rm T},
%\end{equation}
where $\theta_{{\rm R}l}=\frac{s}{\lambda}{\rm cos}(\phi_{{\rm R}l})$ is the directional cosine with carrier wavelength, $\lambda$, antenna spacing, $s~(s\geq\frac{\lambda}{2})$, and angle of arrival (AoA), $\phi_{{\rm R}l}$. $\mathbf{a}_{\rm T}(\theta_{{\rm T}l})$ can be similarly expressed. Then the frequency-domain channel response at the $k$th subcarrier can be represented as
\vspace{-0.2cm}
\begin{align}
\nonumber
&\mathbf{H}[k]
=\sum_{d=0}^{N_{\rm c}-1}\mathbf{H}_de^{-j2\pi\frac{kd}{K}}\\
% &\overset{(a)}{=}\sum_{d=0}^{N_{\rm c}-1}\sum_{l=1}^{N_{\rm p}}\alpha_l\delta(d T_{\rm s}-\tau_l)\mathbf{a}_{\rm R}(\theta_{{\rm R}l})\mathbf{a}_{\rm T}^{\rm H}(\theta_{{\rm T}l})e^{-j2\pi\frac{kd}{K}}\\\nonumber
&=\sum_{l=1}^{N_{\rm p}} \alpha_l \underbrace{\left(\sum_{d=0}^{N_{\rm c}-1}\delta(d-\tau_l)e^{-j2\pi\frac{kd}{K}}\right)}_{\beta_{kl}} \mathbf{a}_{\rm R}(\theta_{{\rm R}l}) \mathbf{a}_{\rm T}^{\rm H}(\theta_{{\rm T}l}).
% &\triangleq\sum_{l=1}^{N_{\rm p}}\alpha_l\beta_{k,l}\mathbf{a}_{\rm R}(\theta_{{\rm R}l})\mathbf{a}_{\rm T}^{\rm H}(\theta_{{\rm T}l}),
\end{align}
\vspace{-0.4cm}

In hybrid massive MIMO illustrated in Fig. 1, the main task is to recover $N_{\rm r}N_{\rm t}$ channel coefficients. Let each training be a transmission of $N_{\rm RFt}$ orthogonal pilots formed by signals at $N_{\rm RFt}$ RF chains. $N_{\rm RFr}N_{\rm RFt}$ coefficients can be estimated at the receiver using $N_{\rm RFt}$ RF chains with each training. Thus, the multiple of, $L\triangleq\lceil \frac{N_{\rm r}N_{\rm t}}{N_{\rm RFr}N_{\rm RFt}}\rceil$, trainings are needed.

Let $M$ be the times of channel use within a coherence time. Then $MN_{\rm RFr}$ observations are utilized to estimate the $N_{\rm r}N_{\rm t}$-dimensional channel vector. For uncorrelated channels, at least $M = \lceil N_{\rm r} N_{\rm t} / N_{\rm {RFr}} \rceil$ channel uses are needed. Let $\mathbf{s}_{m_1}[k], \mathbf{s}_{m_2}[k]\in\mathbb{C}^{N_{\rm RFt}\times1}$ be pilots at the $k$th subcarrier during the $m_1$th and $m_2$th $(m_1, m_2\in\{1,\cdots,M\})$ channel use,
\vspace{-0.2cm}
\begin{equation}
\vspace{-0.2cm}
\mathbf{s}^{\rm H}_{m_1}[k]\mathbf{s}_{m_2}[k] = \begin{cases}
0  &m_1\neq m_2\\
P N_{\rm RFt}           &m_1= m_2,
\end{cases}
\vspace{-0.1cm}
\end{equation}
where $P$ is the pilot power, $m_1 = iN_{\rm RFt}+j_1$ and $m_2 = iN_{\rm RFt}+j_2$ with $i \in \{0,1,\cdots, M_{\rm g}-1\}$ where $M_{\rm g}=M/N_{\rm RFt}~(M_{\rm g}\in \mathbb{N}^+)$ and $1\leq j_1, j_2 \leq N_{\rm RFt}$. The corresponding received signal at the $k$th subcarrier is
\begin{equation}
\vspace{-0.1cm}
\mathbf{r}_m[k]=\mathbf{H}[k]\mathbf{F}_{{\rm A}m}\mathbf{s}_m[k]+\mathbf{v}_m[k],
\vspace{-0.1cm}
\end{equation}
where $\mathbf{F}_{{\rm A}m}\in \mathbb{C}^{N_{\rm t} \times N_{\rm RFt}}$ is the analog precoder and $\mathbf{v}_m[k]\sim\mathcal{CN}(\mathbf{0},\sigma^2_v\mathbf{I}_{N_{\rm r}})$ is the noise vector. At the receiver, an analog estimator, $\mathbf{W}_{{\rm A}m}\in \mathbb{C}^{N_{\rm r} \times N_{\rm RFr}}$, is first implemented. Both $\mathbf{F}_{{\rm A}m}$ and $\mathbf{W}_{{\rm A}m}$ represent the operations of a phase-shifter network connecting the large antenna array to limited RF chains. The phase shifter adjusts only the phase of input signal without changing its amplitude. Thus, each element in $\mathbf{F}_{{\rm A}m}$ and $\mathbf{W}_{{\rm A}m}$ is restricted as a unity-magnitude value. Then after the ADC quantization, the channel coefficients can be estimated via a linear digital estimator, $\mathbf{W}_{{\rm D}m}[k] \in \mathbb{C}^{N_{\rm RFr} \times N_{\rm r}N_{\rm t}}$, incorporating with the former $\mathbf{W}_{{\rm A}m}$. Note that $\mathbf{F}_{{\rm A}m}$ and $\mathbf{W}_{{\rm A}m}$ are performed on the wideband signals in the time domain. Consequently, the same analog processing components apply for all subcarriers, which are thus described as frequency-independent. Digital processing is performed in the frequency domain and thus can be different across subcarriers, which is termed as frequency-selective. The channel estimating problem is
\begin{align}\label{eq:ymk}
&\hat{\mathbf{h}}^{\star}[k]\!=\!{\rm arg} \min_{\mathbf{F}_{{\rm A}m},\mathbf{W}_{{\rm A}m},\mathbf{W}_{{\rm D}m}[k]} \mathbb{E}\Big\{||\hat{\mathbf{h}}[k]\!-\!{\rm vec}(\mathbf{H}[k])||^2_2\Big\},\\\nonumber
&{\rm s.t.}
\begin{cases}
\hat{\mathbf{h}}[k]=\mathbf{W}^{\rm H}_{\rm D}[k]\mathbf{y}[k],~\mathbf{y}[k]\triangleq [\mathbf{y}_1[k], \cdots, \mathbf{y}_M[k]]^{\rm T}\\
\mathbf{W}_{\rm D}[k]\!\triangleq\! [\mathbf{W}_{\rm D1}[k], \cdots, \mathbf{W}_{{\rm D}M}[k]]^{\rm T}\\
\mathbf{y}_m[k] \!=\! \mathcal{Q}\Big(\mathbf{W}_{{\rm A}m}^{\rm H}\mathbf{H}[k]\mathbf{F}_{{\rm A}m}\mathbf{s}_m[k] \!+\! \mathbf{W}_{{\rm A}m}^{\rm H}\mathbf{v}_m[k]\Big)\\
|[\mathbf{F}_{{\rm A}m}]_{ij}|=|[\mathbf{W}_{{\rm A}m}]_{ij}|=1,
\end{cases}
\end{align}
where $\mathcal{Q}(\cdot)$ represents the ADC quantization operation.

\vspace{-0.3cm}
\section{Wideband Hybrid Channel Estimation}
\vspace{-0.1cm}
The wideband channel estimation problem in \eqref{eq:ymk} is challenging. This section focuses on designing the linear hybrid estimators, including $\mathbf{W}_{{\rm A}m}$, $\mathbf{W}_{{\rm D}m}[k]$, and $\mathbf{F}_{{\rm A}m}$.

With limited RF chains, we only have access to a much smaller number of observations per estimation than the number of channel coefficients to be estimated. The design of $\mathbf{W}_{{\rm D}m}[k]$ is frequency-selective while the design of analog ones, $\mathbf{F}_{{\rm A}m}$ and $\mathbf{W}_{{\rm A}m}$, are frequency-independent. Moreover, we have to take into account the ADC quantization $\mathcal{Q}$.

\vspace{-0.3cm}
\subsection{Channel Estimation Formulation}
\vspace{-0.1cm}

From \cite{alkhateeb2014channel}, it is convenient to project channel coefficients onto the angular domain. We use dictionary matrices consisting of ULA response vectors, whose sizes are chosen as $N_{\rm t}$ and $N_{\rm r}$ which denote the numbers of \emph{resolvable angles} at transmitter and receiver, respectively. Specifically, let
\begin{equation}
\mathbf{A}_{\rm t}= \big[\mathbf{a}_{\rm T}(\tilde{\theta}_{{\rm T},1}),\mathbf{a}_{\rm T}(\tilde{\theta}_{{\rm T},2}), \cdots, \mathbf{a}_{\rm T}(\tilde{\theta}_{{\rm T},N_{\rm t}})\big]\in\mathbb{C}^{N_{\rm t}\times N_{\rm t}}
\end{equation}
be the dictionary matrix consisting of columns $\mathbf{a}_{\rm T}(\tilde{\theta}_{{\rm T},p})~(p\in\{1,2,\cdots,N_{\rm t}\})$ with $\tilde{\theta}_{{\rm T},p}\in(-0.5,0.5)$ drawn from a fixed equal interval as $\tilde{\theta}_{{\rm T},p}\triangleq\frac{1}{N_{\rm t}}\big(p-\frac{N_{\rm t}+1}{2}\big)$, and let
\begin{equation}
\mathbf{A}_{\rm r}\!=\! \big[\mathbf{a}_{\rm R}(\tilde{\theta}_{{\rm R},1}),\mathbf{a}_{\rm R}(\tilde{\theta}_{{\rm R},2}), \cdots, \mathbf{a}_{\rm R}(\tilde{\theta}_{{\rm R},N_{\rm r}})\big]\in\mathbb{C}^{N_{\rm r}\times N_{\rm r}}
\end{equation}
be the dictionary matrix composed of $\mathbf{a}_{\rm R}(\tilde{\theta}_{{\rm R},q})~(q\in\{1,2,\cdots,N_{\rm r}\})$ with $\tilde{\theta}_{{\rm R},q}\triangleq\frac{1}{N_{\rm r}}\big(q-\frac{N_{\rm r}+1}{2}\big)$. The equivalent channel matrix after a whole-space projection is defined as
\begin{align}\label{eq:hk2}
\mathbf{H}_{{\rm v}}[k]\triangleq\mathbf{A}_{\rm r}^{\rm H}\mathbf{H}[k]\mathbf{A}_{\rm t}.
\end{align}

Note that the channel model for uniform planar array (UPA) \cite{dailinlong2016JSAC} shares a similar structure as \eqref{eq:hk2}. The only difference lies in that the dictionary matrices for UPA contain an extra quantized angle grid on vertical. Thus the following proposal can also apply for UPA with corresponding subtle changes.

After defining $\mathbf{\Psi}\triangleq\mathbf{A}_{\rm t}^{*}\otimes\mathbf{A}_{\rm r}$ and using matrix vectorization for notational simplicity, we denote $\mathbf{h}_{\rm v}[k]\triangleq{\rm vec}(\mathbf{H}_{\rm v}[k])$ as the equivalent channel coefficient vector to be estimated. From \eqref{eq:hk2} and the unitary properties of $\mathbf{A}_{\rm t}$ and $\mathbf{A}_{\rm r}$, we  write
\begin{align}
\nonumber
\mathbf{y}_m[k]
&\!=\!\mathcal{Q}\Big(\big((\mathbf{s}^{\rm T}_m[k]\mathbf{F}^{\rm T}_{{\rm A}m})\!\otimes\!\mathbf{W}_{{\rm A}m}^{\rm H}\big){\rm vec}\big(\mathbf{H}[k]\big) \!+\! \mathbf{W}_{{\rm A}m}^{\rm H}\mathbf{v}_m[k]\Big)\\
&\!=\!\mathcal{Q}\Big(\mathbf{\Phi}_m[k]\mathbf{\Psi}\mathbf{h}_{\rm v}[k] + \mathbf{e}_m[k]\Big),
\end{align}
where $\mathbf{\Phi}_m[k]\!\triangleq\!\big(\mathbf{s}^{\rm T}_m[k]\mathbf{F}^{\rm T}_{{\rm A}m}\big)\!\otimes\!\mathbf{W}_{{\rm A}m}^{\rm H}$,~$\mathbf{e}_m[k]\!\triangleq\!\mathbf{W}_{{\rm A}m}^{\rm H}\mathbf{v}_m[k]$.

We have to use at least $M$ trainings within a coherence time to estimate a complete $\mathbf{h}_{\rm v}[k]$. Stacking the received signal vectors corresponding to the $M$ trainings, we have
\begin{equation}\label{eq:yk1}
\mathbf{y}[k]\!\triangleq\!\left[\mathbf{y}_1[k],\cdots,\mathbf{y}_M[k]\right]^{\rm T}\!=\!\mathcal{Q}\left(\mathbf{\Phi}[k]\mathbf{\Psi}\mathbf{h}_{\rm v}[k]
\!+\!\mathbf{e}[k]\right),
\end{equation}
where $\mathbf{\Phi}[k]$ and $\mathbf{e}[k]$ are the corresponding stacked vectors of $\mathbf{\Phi}_m[k]$ and $\mathbf{e}_m[k]~(m\in\{ 1,2,\cdots,M\})$, respectively.

\newcounter{TempEqCnt}
\setcounter{TempEqCnt}{\value{equation}}
\setcounter{equation}{23}
\begin{figure*}[hb]
\vspace{-0.2cm}
\hrulefill
\vspace{-0.25cm}
\begin{align}\label{eq:Eee}
&\mathbb{E}\big\{\hat{\mathbf{e}}[k]\hat{\mathbf{e}}^{\rm H}[k]\big\}\nonumber
\overset{(a)}{=}(1-\eta_b)\Big((1-\eta_b)\mathbb{E}\{\mathbf{e}[k]\mathbf{e}^{\rm H}[k]\}+\eta_b {\rm diag}\big(\mathbb{E}\{(\mathbf{\Phi}[k]\mathbf{\Psi}\mathbf{h}_{\rm v}[k]+\mathbf{e}[k])(\mathbf{\Phi}[k]\mathbf{\Psi}\mathbf{h}_{\rm v}[k]+\mathbf{e}[k])^{\rm H}\}\big)\Big)\\
&\overset{(b)}{=}(1-\eta_b)\Big(\sigma_v^2\mathbf{I}_{MN_{\rm RFr}}+\eta_b\sigma_{\rm h}^2{\rm diag}\left(\mathbb{E}\{\mathbf{\Phi}[k]\mathbf{\Psi}\mathbf{\Psi}^{\rm H}\mathbf{\Phi}^{\rm H}[k]\}\right)\Big)\overset{(c)}{=}(1-\eta_b)(\sigma_v^2+\eta_b\sigma_{\rm h}^2PN_{\rm RFt})\mathbf{I}_{MN_{\rm RFr}}\triangleq\sigma^2_{\hat{\mathbf{e}}}\mathbf{I}_{MN_{\rm RFr}}
\end{align}
\vspace{-0.8cm}
\end{figure*}
\setcounter{equation}{\value{equation}}

\setcounter{equation}{24}
\begin{figure*}[hb]
\hrulefill
\vspace{-0.25cm}
\begin{align}\label{eq:Ephipsi}
\nonumber
&{\rm diag}\Big(\mathbb{E}\big\{\mathbf{\Phi}[k]\mathbf{\Psi}\mathbf{\Psi}^{\rm H}\mathbf{\Phi}^{\rm H}[k]\big\}\Big)\\\nonumber
\overset{(a)}{=}&{\rm diag}\Big({\rm diag}\left(\mathbb{E}\{\left(\mathbf{s}^{\rm T}_1[k]\mathbf{F}^{\rm T}_{{\rm A}1}\otimes\mathbf{W}_{{\rm A}1}^{\rm H}\right)\left(\mathbf{F}^*_{{\rm A}1}\mathbf{s}^*_1[k]\otimes\mathbf{W}_{{\rm A}1}\right)\}\right),\cdots,{\rm diag}\left(\mathbb{E}\{\left(\mathbf{s}^{\rm T}_M[k]\mathbf{F}^{\rm T}_{{\rm A}M}\otimes\mathbf{W}_{{\rm A}M}^{\rm H}\right)\left(\mathbf{F}^*_{{\rm A}M}\mathbf{s}^*_M[k]\otimes\mathbf{W}_{{\rm A}M}\right)\}\right)\Big)\\\nonumber
\overset{(b)}{=}&{\rm diag}\Big({\rm diag}\left(\mathbb{E}\{\mathbf{s}^{\rm T}_1[k]\mathbf{s}^*_1[k]\otimes\mathbf{I}_{N_{{\rm RFr}}}\}\right),\cdots,{\rm diag}\left(\mathbb{E}\{\mathbf{s}^{\rm T}_M[k]\mathbf{s}^*_M[k]\otimes\mathbf{I}_{N_{{\rm RFr}}}\}\right)\Big)\\
=&{\rm diag}\Big({\rm diag}\left(PN_{\rm RFt}\mathbf{I}_{N_{{\rm RFr}}}\right),\cdots,{\rm diag}\left(PN_{\rm RFt}\mathbf{I}_{N_{{\rm RFr}}}\right)\Big)=PN_{\rm RFt}\mathbf{I}_{MN_{\rm RFr}}
\end{align}
\end{figure*}
\setcounter{equation}{\value{equation}}

The quantization of ADCs is in general non-linear. Thanks to studies \cite{mezghani2012capacity}\cite{Xu2017WCL} which applied the Bussgang theorem \cite{bussgang1952crosscorrelation} on modelling non-linear quantization, it showed that the output of the non-linear quantizer with Gaussian input can be expressed in closed form by decomposing it into a desired signal component and an uncorrelated quantization distortion, $\mathbf{e}_q$. The output signal after ADC quantization is modelled as
\setcounter{equation}{11}
\begin{align}\label{eq:yk}
\mathbf{y}[k]=(1-\eta_b)\mathbf{\Phi}[k]\mathbf{\Psi}\mathbf{h}_{\rm v}[k]+\hat{\mathbf{e}}[k],
\end{align}
where $\eta_b$ is the distortion factor in terms of the number of quantization bits of ADCs, i.e., $b$, and
\begin{equation}
\hat{\mathbf{e}}[k]=(1-\eta_b)\mathbf{e}[k]+\mathbf{e}_q
\end{equation}
\begin{table}
\center
\caption{Values of quantization distortion factors $\eta_b$ \cite{Max1960TIT}}
\begin{tabular}{|c|c|c|c|c|c|c|c|c|} %l(left)¾Ó×óÏÔʾ r(right)¾ÓÓÒÏÔʾ c¾ÓÖÐÏÔʾ
\hline
$b$&1&2&3&4&5\\
\hline
$\eta_b$&0.3634&0.1175&0.03454&0.009497&0.002499\\
\hline
\end{tabular}
\label{Tab:1}
\vspace{-0.5cm}
\end{table}
is the equivalent noise including both the ADC quantization error and the AWGN. The value of $\eta_b$ is determined by the quantization precision. In the condition of high-precision optimal non-uniform quantizations, $\eta_b$ can be approximately determined by a closed-form expression \cite{mezghani2012capacity}. For a general ADC precision, there is no explicit expression for determining $\eta_b$. While in \cite{Max1960TIT}, typical values of $\eta_b$ corresponding to various precisions are exemplified in TABLE I. Through the digital combiners at the receiver, the estimate in \eqref{eq:ymk} is represented in the virtual-angular domain as
\begin{equation}\label{eq:hvk}
\hat{\mathbf{h}}_{\rm v}[k]\!=\!\mathbf{W}_{{\rm D}}^{\rm H}[k]\Big((1-\eta_b)\mathbf{\Phi}[k]\mathbf{\Psi}\mathbf{h}_{\rm v}[k]\!+\!\hat{\mathbf{e}}[k]\Big),
\end{equation}
where $\mathbf{W}_{\rm D}[k]\!\triangleq\!\big[\mathbf{W}_{\rm D1}[k],\cdots,\mathbf{W}_{{\rm D}M}[k]\big]^{\rm T}$. In \eqref{eq:hvk}, the channel estimation problem is converted into finding proper $\mathbf{W}_{\rm D}[k]$ and $\mathbf{\Phi}[k]$, which requires careful evaluation of the randomness of terms in \eqref{eq:hvk} and will be discussed subsequently.

Without \emph{priori} channel statistics or exploiting channel sparsity, we have to solve the estimation problem in \eqref{eq:hvk} with $N_{\rm r}N_{\rm t}$ coefficients in $\hat{\mathbf{h}}_{\rm v}[k]$, or equivalently $\mathbf{h}_{\rm v}[k]$. When \emph{a priori} channel sparsity, $N_{\rm v}~(N_{\rm v}\!\ll\! N_{\rm r}N_{\rm t})$, is exploited, problem \eqref{eq:hvk} is rewritten to incorporate compressed sensing (CS) techniques, e.g., OMP, for complexity and pilot overhead reduction. With OMP, we use a uniform selective matrix to pick the dominant coefficients for estimation. It yields
\begin{align}\label{eq:hNZvk}
\mathbf{h}^{\rm NZ}_{\rm v}[k]\!=\!\mathbf{P}^{\rm T}_{\rm v}[k]\mathbf{h}_{\rm v}[k]\!\triangleq\!\big[\mathbf{e}_{\pi(1)}, \cdots, \mathbf{e}_{\pi(N_{\rm v})}\big]^{\rm T}\mathbf{h}_{\rm v}[k],
\end{align}
where $\mathbf{e}_{\pi(i)}~(\pi(i)\in\{1,2,\cdots, N_{\rm r}N_{\rm t}\})$ is a selecting vector with the $\pi(i)$th element being 1. According to \eqref{eq:hNZvk}, we only need to estimate $N_{\rm v}$ non-zero channel coefficients. Potentially, $M$ can be at most reduced to $\big\lceil N_{\rm v}/N_{\rm RFr}\big\rceil$.

The determination of $\mathbf{P}_{\rm v}[k]$ depends on whether we have \emph{priori} channel sparsity information. If there is no \emph{priori} sparsity information, we simply have $\mathbf{P}_{\rm v}[k]=\mathbf{I}_{N_{\rm r}N_{\rm t}}$. For sparse channels without knowing $N_{\rm v}$, we apply OMP to obtain locations of $N_{\rm v}$ dominant channel coefficients. In this case, the selecting matrix, $\mathbf{P}_{\rm v}[k]$, can be elaborated as in Appendix A.

\vspace{-0.5cm}
\subsection{Linear Channel Estimator Optimization}
\vspace{-0.1cm}
The channel estimation problem in \eqref{eq:hNZvk} and \eqref{eq:hvk} is to design $\mathbf{W}_{\rm D}[k]$ and $\mathbf{\Phi}[k]$, equivalently $\mathbf{F}_{{\rm A}m}$ and $\mathbf{W}_{{\rm A}m}$, to estimate $\mathbf{h}^{\rm NZ}_{\rm v}[k]$. By substituting \eqref{eq:hNZvk} into \eqref{eq:yk}, we have
\begin{align}\label{eq:Byk}
\nonumber
\mathbf{y}[k]
&=(1-\eta_b)\big(\mathbf{\Phi}[k]\mathbf{\Psi}\mathbf{P}_{\rm v}[k]\big)\mathbf{h}^{\rm NZ}_{\rm v}[k]+\hat{\mathbf{e}}[k]\\
&\triangleq(1-\eta_b)\mathbf{\Omega}[k]\mathbf{h}^{\rm NZ}_{\rm v}[k]+\hat{\mathbf{e}}[k].
\end{align}
Without any \emph{priori} channel directivity information, it is reasonable to apply isotropic pilot directions, which corresponds to i.i.d. Gaussian $\mathbf{F}_{{\rm A}m}$ and corresponding $\mathbf{W}_{{\rm A}m}$. Under the hybrid architecture, however, generating i.i.d. Gaussian matrix is infeasible due to analog hardware limitation. Alternatively, we choose that $\mathbf{F}_{{\rm A}m}$ and $\mathbf{W}_{{\rm A}m}$ have phases drawn uniformly from $[0,2\pi)$ via phase shifters with unimodular constraints. In practice, we can generate fixed analog processing matrices corresponding to the uniform distribution. These fixed matrices form a codebook, in which each matrix is a codeword. The codebook can be predetermined and shared by both sides.

Our goal remains to optimize $\mathbf{W}_{\rm D}[k]\in\mathbb{C}^{MN_{\rm RFr} \times N_{\rm v}}~(N_{\rm v}\leq N_{\rm r}N_{\rm t})$ to estimate $\mathbf{h}^{\rm NZ}_{\rm v}[k]$ from $\mathbf{y}[k]$ in \eqref{eq:Byk}, i.e.,
\begin{equation}\label{eq:BhNZv}
\hat{\mathbf{h}}^{\rm NZ}_{\rm v}[k]\!=\!(1-\eta_b)\mathbf{W}_{\rm D}^{\rm H}[k]\mathbf{\Omega}[k]\mathbf{h}^{\rm NZ}_{\rm v}[k]\!+\!\mathbf{W}_{\rm D}^{\rm H}[k]\hat{\mathbf{e}}[k].
\end{equation}

From \eqref{eq:BhNZv} and Appendix B, the optimal digital estimator in terms of MMSE is derived as
\begin{equation}\label{eq:optW}
\mathbf{W}^{\rm \star}_{\rm D}[k]\!=\!\frac{\mathbf{\Omega}[k]}{1\!-\!\eta_b}\left(\mathbf{\Omega}^{\rm H}[k]\mathbf{\Omega}[k]\!+\!\frac{\sigma^2_{\hat{\mathbf{e}}}\mathbf{I}_{N_{\rm v}}}{(1-\eta_b)^2\sigma^2_{\rm h}}\right)^{-1},
\end{equation}
where $\sigma^2_{\hat{\mathbf{e}}}$ is the variance of each element of the effective noise vector, $\hat{\mathbf{e}}[k]$, and $\sigma^2_{\rm h}$ is the large-scale fading factor of $\mathbf{h}_{\rm v}[k]$.

To this end, we have obtained the channel estimate as
\begin{equation}
\hat{\mathbf{h}}[k]\!=\!\big(\mathbf{A}_{\rm t}^{*}\otimes\mathbf{A}_{\rm r}\big)\big(\mathbf{P}^{\rm T}_{\rm v}[k]\big)^{-1}\big(\mathbf{W}_{\rm D}^{\rm \star}[k]\big)^{\rm H}\mathbf{y}[k],
\end{equation}
where $\big(\mathbf{P}^{\rm T}_{\rm v}[k]\big)^{-1}$ is the inverse operation of $\mathbf{P}^{\rm T}_{\rm v}[k]$.

For completeness, if correlations across the frequency domain are considered, pilots can be inserted every few subcarriers. The minimum required length for channel training is $\lceil K\Delta f/B_{\rm c}\rceil$ in the frequency domain, where $B_{\rm c}$ and $\Delta f$ are the coherence bandwidth and subcarrier spacing, respectively.

\setcounter{equation}{26}
\begin{figure*}[ht]
\begin{align}\label{eq:MSE}
\nonumber
&{\rm MSE}
=\mathbb{E}\Big\{\big|\big|\hat{\mathbf{h}}^{\rm NZ}_{\rm v}[k]-\mathbf{h}^{\rm NZ}_{\rm v}[k]\big|\big|_2^2\Big\}=\mathbb{E}\Big\{\big|\big|((1-\eta_b)\mathbf{W}_{\rm D}^{\rm H}[k]\mathbf{\Omega}[k]-\mathbf{I}_{N_{\rm v}})\mathbf{h}^{\rm NZ}_{\rm v}[k]+\mathbf{W}_{\rm D}^{\rm H}[k]\hat{\mathbf{e}}[k]\big|\big|_2^2\Big\}\\
=&(1\!-\!\eta_b)\sigma^2_{\rm h}\Big((1\!-\!\eta_b){\rm Tr}\big(\mathbf{W}_{\rm D}^{\rm H}[k]\mathbf{\Omega}[k]\mathbf{\Omega}^{\rm H}[k]\mathbf{W}_{\rm D}[k]\big)\!-\!{\rm Tr}\big(\mathbf{W}_{\rm D}[k]\mathbf{\Omega}^{\rm H}[k]\!+\!\mathbf{\Omega}[k]\mathbf{W}_{\rm D}^{\rm H}[k]\big)\Big)\!+\!\sigma^2_{\rm h}N_{\rm v}\!+\!\sigma^2_{\hat{\mathbf{e}}}{\rm Tr}\big(\mathbf{W}_{\rm D}^{\rm H}[k]\mathbf{W}_{\rm D}[k]\big)
\end{align}
\hrulefill
\end{figure*}
\setcounter{equation}{\value{equation}}

\setcounter{equation}{27}
\begin{figure*}[ht]
\vspace{-0.6cm}
\begin{align}\label{eq:partial}
\frac{\partial ({\rm MSE})}{\partial\mathbf{W}_{\rm D}^{*}[k]}
=(1-\eta_b)^2\sigma^2_{\rm h}\mathbf{\Omega}[k]\mathbf{\Omega}^{\rm H}[k]\mathbf{W}_{\rm D}[k]-(1-\eta_b)\sigma^2_{\rm h}\mathbf{\Omega}[k]+\sigma^2_{\hat{\mathbf{e}}}\mathbf{W}_{\rm D}[k]
\end{align}
\hrulefill
\vspace{-0.6cm}
\end{figure*}
\setcounter{equation}{\value{equation}}

\begin{figure}
\centering
\begin{minipage}[t]{0.49\linewidth}
\includegraphics[width=1.85in]{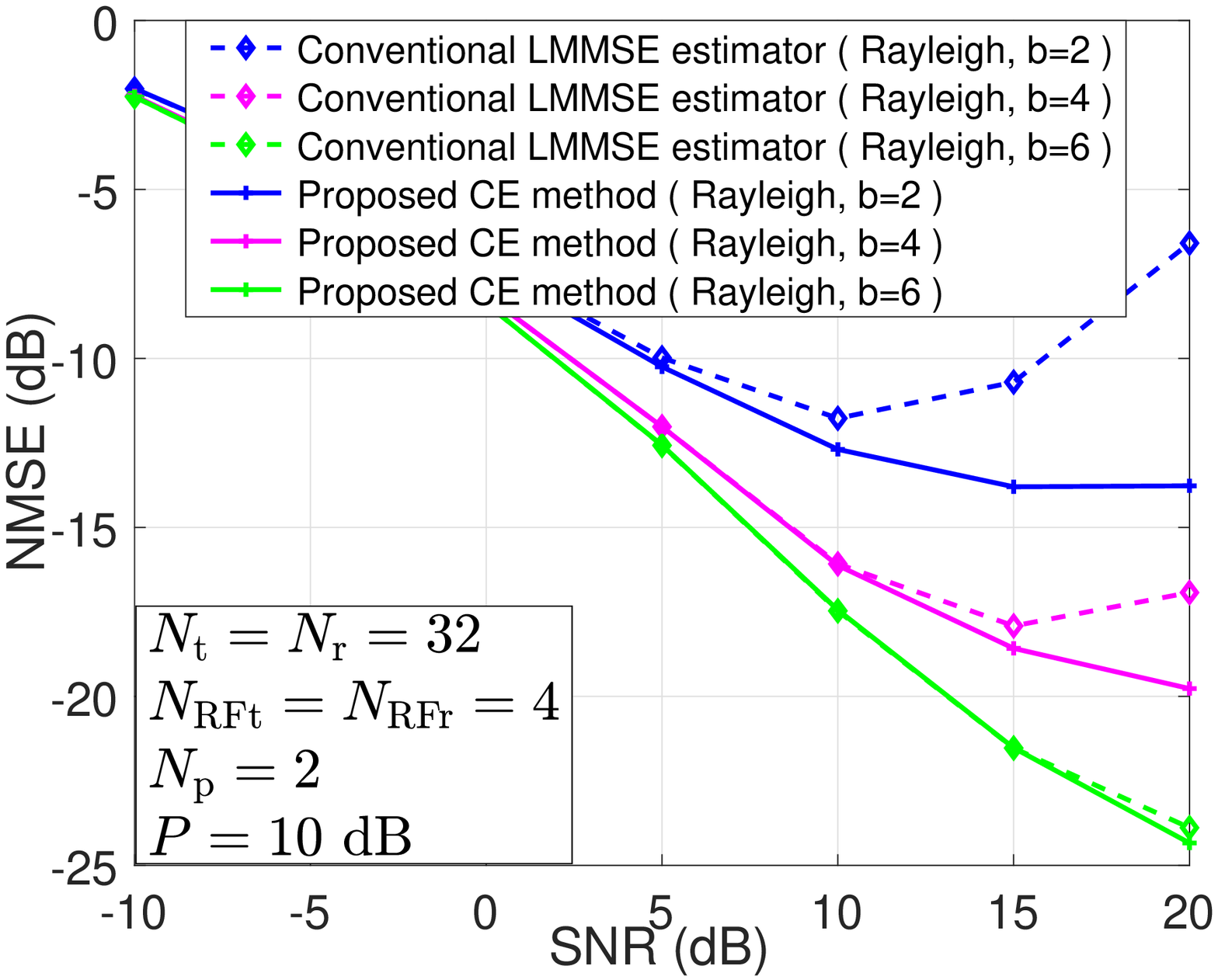}
\vspace{-0.6cm}
\caption{NMSE for Rayleigh channels.}
\label{Fig:NMSE1}
\end{minipage}
\begin{minipage}[t]{0.49\linewidth}
\centering
\includegraphics[width=1.85in]{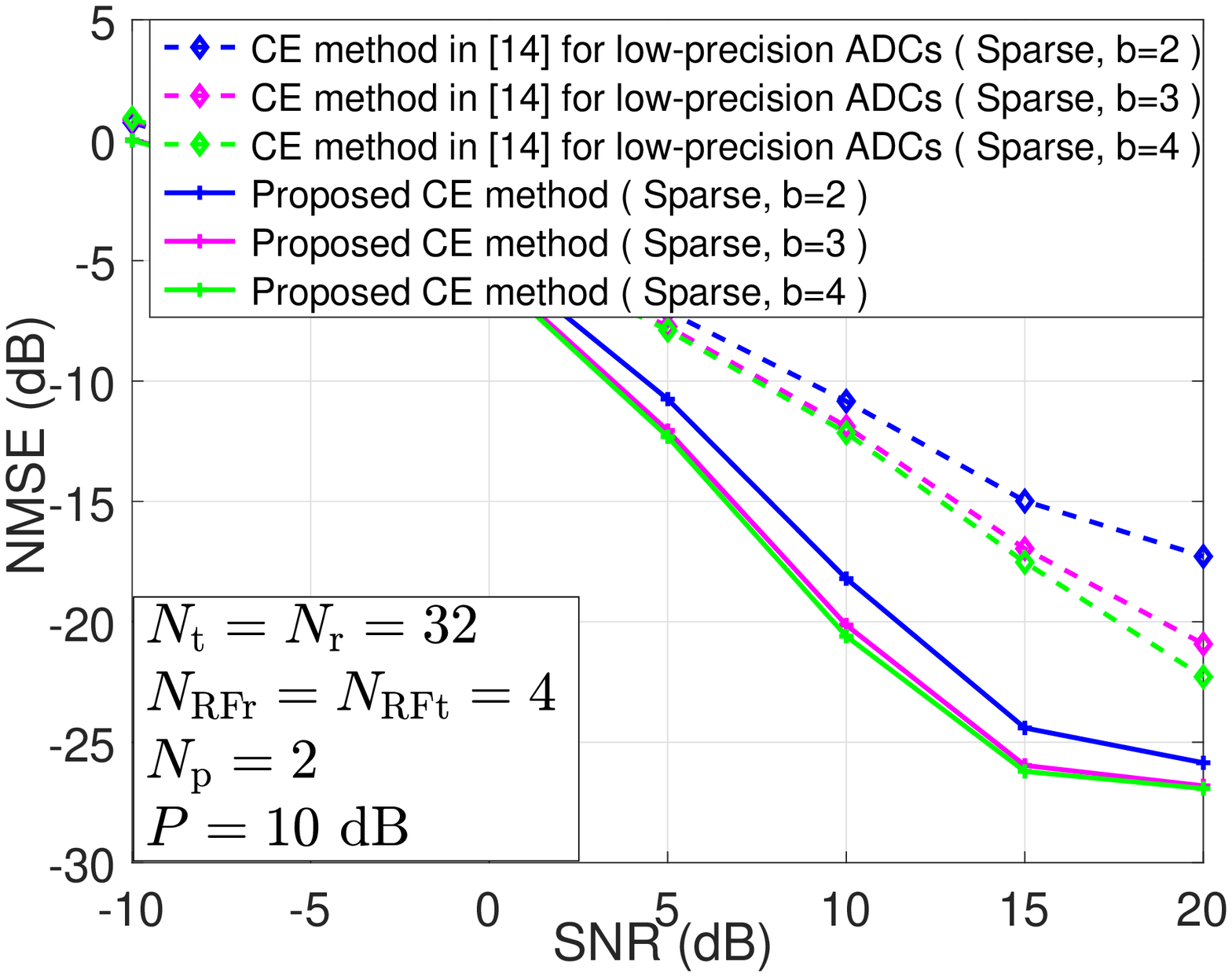}
\vspace{-0.6cm}
\caption{NMSE for sparse channels.}
\label{Fig:NMSE2}
\end{minipage}
\vspace{-0.5cm}
\end{figure}

\vspace{-0.3cm}
\section{Simulation Results}
\vspace{-0.1cm}

The performance is shown in Fig. 2 and Fig. 3 using the model in \eqref{eq:yk}, where the normalized mean square error (NMSE) is defined as
%\setcounter{equation}{18}
%\begin{equation}
%{\rm NMSE}=\mathbb{E}\Bigg\{\frac{\big|\big|\hat{\mathbf{h}}^{\rm NZ}_{\rm v}[k]-\mathbf{h}^{\rm NZ}_{\rm v}[k]\big|\big|_2^2}{\big|\big|\mathbf{h}^{\rm NZ}_{\rm v}[k]\big|\big|_2^2}\Bigg\}.
%\end{equation}
${\rm NMSE}=\mathbb{E}\Big\{\big|\big|\hat{\mathbf{h}}^{\rm NZ}_{\rm v}[k]-\mathbf{h}^{\rm NZ}_{\rm v}[k]\big|\big|_2^2\Big/\big|\big|\mathbf{h}^{\rm NZ}_{\rm v}[k]\big|\big|_2^2\Big\}.$

The figures show the comparison between our proposed method and the conventional ones, e.g., \cite{Jianhua2018TSP}, under Rayleigh and sparse channels, respectively. It shows that our method outperforms the conventional ones under various SNRs and ADC precisions for both Rayleigh and sparse channels, especially at median-to-high SNRs. For the conventional LMMSE estimator in Fig. \ref{Fig:NMSE1}, It is interesting that the NMSE grows with SNR when low-precision ADCs are utilized. This is caused by the ADC-quantization-noise-amplification effect at low signal-to-quantization-noise ratio. For the proposed method exploiting sparsity of mmWave channels, OMP is used to further improve the performance as shown in Fig. \ref{Fig:NMSE2}.

\vspace{-0.3cm}
\section{Conclusion}
\vspace{-0.1cm}
We have proposed a general estimator under arbitrary channel statistics for wideband mmWave MIMO with hybrid architecture and low-precision ADCs.

\vspace{-0.3cm}
\appendices
\section{OMP-based CS in \eqref{eq:hNZvk}}
\vspace{-0.1cm}
If sparsity presents in $\mathbf{h}_{\rm v}[k]$, we exploit OMP to estimate the locations of the dominant channel coefficients, i.e., $\mathbf{P}_{\rm v}[k]$. The problem is formulated as
\setcounter{equation}{19}
\begin{align}
&\mathbf{h}_{\rm v}^{\star}[k]=\arg\min_{\mathbf{h}_{\rm v}[k]}||\mathbf{h}_{\rm v}[k]||_1\\\nonumber
{\rm s.t.}~||&\mathbf{y}[k]-(1-\eta_b)\mathbf{\Phi}[k]\mathbf{\Psi}\mathbf{h}_{\rm v}^{\star}[k]||_2\leq\epsilon,
\end{align}
where $\epsilon$ is the stopping threshold. A suitable choice for the threshold $\epsilon$ is the noise variance, i.e., $\epsilon\triangleq\mathbb{E}\{\hat{\mathbf{e}}^{\rm H}[k]\hat{\mathbf{e}}[k]\}$ which for our specific problem is evaluated in \eqref{eq:sigmae} of Appendix~B.

\vspace{-0.3cm}
\section{LMMSE Estimator $\mathbf{W}^{\rm \star}_{\rm D}[k]$}
\vspace{-0.1cm}
The optimization of $\mathbf{W}_{\rm D}[k]$ for estimation in terms of MSE requires the evaluation of some covariance matrices. We present the useful results as follows:
\begin{mylemma}
As $N_{\rm r}, N_{\rm t}\rightarrow \infty$, the variance of $\hat{\mathbf{e}}[k]$ equals
\begin{equation}\label{eq:sigmae}
\mathbb{E}\big\{\hat{\mathbf{e}}^{\rm H}[k]\hat{\mathbf{e}}[k]\big\}=MN_{\rm RFr}(1-\eta_b)(\sigma_v^2+\eta_b\sigma_{\rm h}^2PN_{\rm RFt}).
\end{equation}
\end{mylemma}
\begin{proof}
Given $\mathbf{F}_{{\rm A}m}$ and $\mathbf{W}_{{\rm A}m}$ with phases chosen uniformly from $[0,2\pi)$, we have $[\mathbf{F}_{{\rm A}m}]_{ij}=\frac{1}{\sqrt{N_{\rm t}}}e^{j\phi_{ij}}$ and $[\mathbf{W}_{{\rm A}m}]_{ij}=\frac{1}{\sqrt{N_{\rm r}}}e^{j\psi_{ij}}$ with $\phi_{ij}, \psi_{ij}\sim U[0,2\pi)$. For large $N_{\rm t}$ and $N_{\rm r}$, applying the Law of Large Numbers, we have
\begin{equation}\label{eq:EFAm}
\mathbb{E}\{\mathbf{F}^{\rm H}_{{\rm A}m}\mathbf{F}_{{\rm A}m}\}=\mathbf{I}_{N_{{\rm RFt}}},~
\mathbb{E}\{\mathbf{W}^{\rm H}_{{\rm A}m}\mathbf{W}_{{\rm A}m}\}=\mathbf{I}_{N_{{\rm RFr}}}.
\end{equation}
Using \eqref{eq:EFAm} and from $\mathbf{e}[k]$ in \eqref{eq:yk1}, we obtain
\begin{equation}\label{eq:EeeH}
\mathbb{E}\{\mathbf{e}[k]\mathbf{e}^{\rm H}[k]\}=\sigma_v^2\mathbf{I}_{MN_{\rm RFr}}.
\end{equation}

Assume that $\mathbf{h}_{\rm v}[k]$ satisfies $\mathbb{E}\{\mathbf{h}_{\rm v}[k]\mathbf{h}^{\rm H}_{\rm v}[k]\}=\sigma^2_{\rm h}\mathbf{I}_{N_{\rm r}N_{\rm t}}$ with $\sigma^2_{\rm h}$ known in advance. Then the autocorrelation matrix of $\hat{\mathbf{e}}[k]$ is calculated in \eqref{eq:Eee} where $(a)$ uses \cite[eq.~(30)]{mezghani2012capacity}, $(b)$ uses \eqref{eq:EeeH}, and $(c)$ uses \eqref{eq:Ephipsi} which characterizes the value of the matrix ${\rm diag}\left(\mathbb{E}\{\mathbf{\Phi}[k]\mathbf{\Psi}\mathbf{\Psi}^{\rm H}\mathbf{\Phi}^{\rm H}[k]\}\right)$. Note that in \eqref{eq:Ephipsi}, $(a)$ uses
the property of $\mathbf{\Psi}$ that
\setcounter{equation}{25}
\begin{equation}\label{eq:phiphiH}
\mathbf{\Psi}\mathbf{\Psi}^{\rm H}\!=\!\mathbf{A}_{\rm t}^{*}\mathbf{A}_{\rm t}^{\rm T}\!\otimes\!\big(\mathbf{A}_{\rm r}\mathbf{A}_{\rm r}^{\rm H}\big)\!=\!\mathbf{I}_{N_{\rm r}N_{\rm t}},
\end{equation}and $(b)$ uses \eqref{eq:EFAm}.
\end{proof}

Using Lemma 1 and from \eqref{eq:BhNZv}, we are ready to give the MSE of $\mathbf{h}^{\rm NZ}_{\rm v}[k]$ in \eqref{eq:MSE}. By forcing $\frac{\partial({\rm MSE})}{\partial\mathbf{W}_{\rm D}^{*}[k]}$ in \eqref{eq:partial} to zero, we obtain the optimal $\mathbf{W}^{\rm \star}_{\rm D}[k]$ as desired in \eqref{eq:optW}.

\end{document}